\documentclass[12pt]{article}
\usepackage[margin=1in]{geometry}
\usepackage{amsmath,amssymb,amsthm,mathtools}
\usepackage{bm}
\usepackage{graphicx}
\usepackage{enumitem}
\usepackage[colorlinks=true,linkcolor=blue,citecolor=blue,urlcolor=blue]{hyperref}
\usepackage{setspace}

\newtheorem{theorem}{Theorem}
\newtheorem{proposition}[theorem]{Proposition}
\newtheorem{lemma}[theorem]{Lemma}
\newtheorem{corollary}[theorem]{Corollary}
\theoremstyle{definition}
\newtheorem{definition}[theorem]{Definition}

\theoremstyle{remark}
\newtheorem{remark}[theorem]{Remark}

\DeclareMathOperator{\spec}{spec}
\newcommand{\Lind}{\mathcal{L}}
\newcommand{\Hilb}{\mathcal{H}}
\newcommand{\id}{\mathbb{I}}
\newcommand{\Tr}{\operatorname{Tr}}
\newcommand{\sig}{\sigma}
\newcommand{\gap}{\Delta}

\newcommand{\inner}[2]{\langle #1, #2\rangle}

\newcommand{\Qch}{Q}

\newcommand{\normsig}[1]{\|#1\|_\sig}
\newcommand{\tracenorm}[1]{\|#1\|_1}

\usepackage{titling}
\pretitle{\begin{center}\fontsize{16}{18}\selectfont\bfseries}
\posttitle{\par\end{center}}
\title{Symmetry-initialized quantum Gibbs sampling:\\[0.2cm] a non-Abelian asymmetry cascade
}
\author{\sc Soo-Jong Rey\\[0.2cm]
{\sl Kwangwoon University}\\
{\sl Seoul, Korea}}
\date{\tt sjrey@kw.ac.kr}

\begin{document}
\doublespacing
\maketitle
\begin{abstract}
For a quantum Gibbs sampler whose mixing bottleneck is a weakly broken symmetry, I show that the correct initialization is determined by representation theory. I first prove a general speedup-versus-prefactor dichotomy: by exactly eliminating slow-mode overlap, I convert a nominal prefactor reduction into a fundamental, asymptotic acceleration of the system's mixing time. I then show that when the bottleneck is a weakly broken symmetry, the otherwise exponentially expensive bottleneck eigenvector is the symmetry charge.For a non-Abelian group $G$, I prove that the correct initialization target is the group-averaging asymmetry. Projecting this asymmetry out requires a $G$-invariant input. Partial, subgroup-invariant inputs produce a cascade of mixing-time speedups indexed by the subgroup lattice. Conversely, matching first moments alone is provably insufficient.The asymmetry is a directly measurable initialization target. I verify these results analytically and numerically in an $SU(2)$ Davies sampler where total-spin multiplets constitute the slow modes. The predicted speedup cascade emerges with relaxation rates scaling linearly with the symmetry-breaking parameter. Finally, the model maps the boundaries of the asymmetry target, illustrating where the separate matching of conserved logical data becomes necessary.
\end{abstract}

\tableofcontents

\section{Introduction}
A warm start is the cheapest way to accelerate a Markov chain. It changes the input, not the chain.
For a quantum Gibbs sampler $\Lind$ with target $\sig$, the cost of preparing $\sig$ from a
generic input is set by the slowest mode. A warm start placing no weight on that mode
removes it. The obstruction is that finding the slow mode is in general as hard as diagonalizing
$\Lind$. In this paper, I show that a symmetry resolves this at no cost, and that a \emph{non-Abelian} symmetry produces a structure with no Abelian counterpart.

I prove a speedup-versus-prefactor dichotomy (Theorem~\ref{prop:dichotomy}) that exact cancellation of the slow-mode projection transforms a nominal prefactor improvement into a fundamental increase in the mixing rate. When the bottleneck is a weakly broken symmetry, the otherwise exponentially expensive slow eigenvector collapses into a known a priori symmetry charge. The warm start then reduces to matching a single thermal expectation (Theorem~\ref{prop:symmetry}, Corollary~\ref{cor:warmalg}). For a compact non-Abelian group \(G\), I prove that the correct initialization target is the Haar-invariant asymmetry \(\mathsf{A}_{G}\). By Schur's Lemma, the slow manifold splits into irreducible representations (Lemma~\ref{lem:schur}). Completely projecting out \(\mathsf{A}_{G}\) requires a \(G\)-invariant input, which yields the maximum mixing speedup. Alternatively, subgroup-invariant inputs produce a cascade of partial speedups indexed directly by the subgroup lattice (Theorem~\ref{prop:nonabelian}). Matching first moments alone is provably
insufficient (Corollary~\ref{cor:firstmoments}).

The irreducible representation splitting is determined by applying Schur's lemma to the symmetry-covariant perturbations within Kato degenerate perturbation theory. Operationally, this maps directly to the resource theory of asymmetry \cite{MarvianSpekkens}. The Haar-average channel projects onto the invariant subalgebra, and the asymmetry \(\mathsf{A}_{G}\) measures the relative entropy to this sector. The quantum Pinsker inequality then directly bounds the remaining slow-mode residual by the initial asymmetry.

The covariance of a detailed-balanced generator forces its relaxation spectrum to block-diagonalize by irreducible representation via Schur's lemma. While this backbone is a standard tool, Basso, Bergamaschi, Lin, Ragone, and Stubbs concurrently applied it to $S_n$-symmetric Davies dynamics in $SU(2)$ Heisenberg Gibbs sampler~\cite{LinConcurrent}.\footnote{I thank Thiago Bergamaschi for informing me of the work.}. Here, I deploy this framework as an initialization principle. The exact-projection condition $a_1(\rho_0)=0$ originates with Lu and Raz \cite{LuRaz,KlichRaz}. \cite{LuRaz,KlichRaz}. Rather than canceling modes by explicit per-instance construction as in Beato and Teza \cite{BeatoTeza}, the symmetry directly identifies the target mode and reduces the cancellation to a group-invariance condition on the input. For a single Abelian charge, this recovers the logical-sector initialization of \cite{BergamaschiGheissariLiu}. The extension to non-Abelian multiplets, the group-average asymmetry as the target, and the subgroup-lattice cascade constitute the new content of this work. 

\section{Preliminaries}\label{sec:prelim}

I collect the standard objects, fixing conventions; the reader fluent in dissipative Gibbs samplers may skip to Section~\ref{sec:setup}.

For a Hamiltonian $H$ on a finite-dimensional Hilbert space $\Hilb$ and an inverse temperature
$\beta$, the Gibbs state is $\sig=e^{-\beta H}/Z$, $Z=\Tr e^{-\beta H}$. A \emph{Gibbs sampler} is a
Markovian (Lindblad) generator $\Lind$ acting on density operators,
\begin{equation}
\Lind(\rho)=-i[H_{\mathrm{LS}},\rho]+\sum_{a}\Big(L_a\rho L_a^\dagger-\tfrac12\{L_a^\dagger L_a,\rho\}\Big),
\label{eq:lindblad}
\end{equation}
engineered so that $\Lind(\sig)=0$ and $\rho_t=e^{t\Lind}(\rho_0)\to\sig$ for every input. The
canonical construction is the Davies generator \cite{CKG,DingLiLin,Rouze}: from a set of system
coupling operators $\{S_\alpha\}$ one forms the Bohr-frequency components
$S_\alpha(\omega)=\sum_{\varepsilon'-\varepsilon=\omega}\Pi_\varepsilon S_\alpha\Pi_{\varepsilon'}$,
where $\Pi_\varepsilon$ projects onto the energy-$\varepsilon$ eigenspace of $H$, and takes the
jumps $L_{\alpha,\omega}=\sqrt{\gamma_\alpha(\omega)}\,S_\alpha(\omega)$ with bath rates obeying the
Kubo--Martin--Schwinger (KMS) condition $\gamma_\alpha(-\omega)=e^{-\beta\omega}\gamma_\alpha(\omega)$.
Modern quasi-local constructions \cite{CKG} produce such generators with $O(1)$-local jumps. This is the precise framework I consider here. The subsequent derivations rely exclusively on the two structural properties stated below. 

I equip the space of operators with the {GNS--KMS inner product}:
\begin{equation}
\inner{X}{Y}_\sig=\Tr\!\big(\sig^{1/2}X^\dagger\sig^{1/2}Y\big).
\label{eq:kms}
\end{equation}
A sampler is \emph{reversible} (satisfies quantum detailed balance) if its generator is self-adjoint with respect to the inner product \eqref{eq:kms}: \(\inner{X}{\mathcal{L}^\dagger Y}_\sig = \inner{\mathcal{L}^\dagger X}{Y}_\sig\).  
Here, $\Lind^\dagger$ denotes the Heisenberg-picture generator, defined via the trace duality \( \Tr(Y^\dagger \mathcal{L}(\rho)) = \Tr((\mathcal{L}^\dagger Y)^\dagger \rho) \).
Davies generators are reversible. This self-adjointness has three main consequences. First, the spectrum of \(\mathcal{L}\) is real and nonpositive. Second, eigenvectors with distinct eigenvalues are orthogonal under the \(\inner{\cdot}{\cdot}_\sig\) inner product. Third, the Heisenberg and Schr\"dinger pictures share the same spectral data, so I can compute in whichever picture is convenient. I denote the superoperator adjoint as \(\mathcal{L}^{\dag }\) throughout. This operator satisfies \(\mathcal{L}^\dagger\mathbb{I}=0 \Leftrightarrow \mathcal{L}(\sigma)=0\).

 Davies generators are reversible.
Self-adjointness has three consequences I use repeatedly: the spectrum of $\Lind$ is real and
nonpositive; eigenvectors belonging to distinct rates are $\inner{\cdot}{\cdot}_\sig$-orthogonal;
and the Heisenberg and Schr\"odinger pictures carry the same spectral data, so I may compute with
whichever is convenient. I write the (super)operator adjoint as $\Lind^\dagger$ throughout and note
$\Lind^\dagger\mathbb{I}=0\Leftrightarrow\Lind(\sig)=0$.

I analyze the relationship between the spectral gap, the mixing time, and the system's slow modes. First, I order the relaxation rates defined by the eigenvalues of $- \Lind$ as:
\[
0=g_0<g_1\le g_2\le\cdots,
\]
where $g_0=0$ corresponds to the stationary mode with \(r_0=\sig\) and $\ell_0=\mathbb{I}$. 
The spectral gap is $g_1$.
The mixing time is evaluated in the standard trace norm,
\[
t_{\mathrm{mix}}(\rho_0;\epsilon)=\inf\{t:\tracenorm{\rho_t-\sig}\le\epsilon\}.
\]
While spectral theory lives in the weighted space where the
generator is self-adjoint, I bridge the two norms via the $\chi^2$-divergence: $\chi^2(\rho\|\sig):=\Tr\!\big[(\rho-\sig)\,\sig^{-1/2}(\rho-\sig)\,\sig^{-1/2}\big]$. This definition forms the
state-side dual of Eq.\eqref{eq:kms}. The $\sig^{-1/2}$ weighting is to up-weight low-probability discrepancies, as in the classical $\chi^2$. 
By Cauchy--Schwarz, the trace norm is bounded by the divergence as $\tracenorm{\rho-\sig}^2\le\chi^2(\rho\Vert{}\sig)$. Under a KMS-reversible semigroup, this $\chi^2$ quantity contracts as \(e^{-2g_{1}t}\) \cite{TKRWV}. For any initial input, the divergence is bounded by \(\chi^2(\rho_0\Vert{}\sig)\le1/\sig_{\min}-1\), where \(\sig_{\min}\) is the smallest eigenvalue of $\sig$.

These relations yield the standard mixing time bound: $ t_{\mathrm{mix}}\le g_1^{-1}\big[\log(1/\epsilon)+\tfrac12\log(1/\sig_{\min})\big]. $ Consequently, relaxation rates transfer between the two norms unchanged. Only the prefactors acquire the state-independent additive penalty $\tfrac12\log(1/\sig_{\min})$. I therefore execute all spectral arguments within the KMS/$\chi^2$-framework and quote the final mixing times in trace norm. The core objective of this paper is to determine when and how the leading rate can be upgraded from $g_1$ to $g_2$ purely through the choice of $\rho_0$.

I now recall the formal definition of a symmetry on a Lindbladian. Let $g\mapsto U_g$ be a unitary representation of a compact group \(G\) on \(\mathcal{H}\), with the conjugation action defined as \(\mathcal{U}_g(X)=U_gXU_g^\dagger\). The generator possesses a \emph{weak symmetry}, or is \emph{\(G\)-covariant}, if $[\mathcal{U}_g,\Lind]=0$ for all \(g\). Under this condition, symmetry-related inputs evolve identically and $\sig$ remains \(G\)-invariant. By contrast, the generator possesses a \emph{strong symmetry} if every jump operator commutes directly with \(U_{g}\). Equivalently, every \(G\)-tensor operator constructed from the conserved charges is annihilated by $\Lind^\dagger$. Consequently, a strong symmetry implies a strict conservation law, whereas a weak symmetry does not.This distinction is central to the upcoming derivations. The unperturbed sampler $\Lind_0$ carries a strong symmetry, yielding a conservation law and exact zero modes. The perturbation retains only the weak symmetry. This covariance lifts those zero modes to small, non-zero relaxation rates while keeping the irreducible representation structure completely intact.

One must consider whether strong symmetries can hold in practice. For a generic and efficient physical sampler, the set of jump operators spans the full local algebra. Because nothing non-trivial commutes with every jump operator, no exact strong symmetry exists on the physical device.However, this property applies to the analytical decomposition of the generator rather than the physical hardware. The unperturbed generator \(\mathcal{L}_{0}\) represents the analytically solvable part of the system. This sector is constructed either by restricting the coupling set to \(G\)-scalars or by isolating the \(G\)-symmetric coupling content within a given sampler. The remaining terms constitute the covariant perturbation \(\eta\mathcal{L}_1\).The subsequent derivations require only two precise conditions. First, the kernel of \(\mathcal{L}_{0}^{\dag }\) restricted to the traceless subspace \(\mathbb{I}^{\perp }\) must consist entirely of the charge multiplets. These multiplets must be separated from the bulk spectrum by an \(O(1)\) spectral gap. Second, both \(\mathcal{L}_{0}\) and \(\mathcal{L}_{1}\) must remain \(G\)-covariant. A strong symmetry of \(\mathcal{L}_{0}\) is simply the physical mechanism that guarantees this spectral structure, rather than a restrictive assumption on the physical jump set.

The physical validity of a strong symmetry warrants careful consideration. For a generic, efficient sampler, the jump operators span the full local algebra. As no non-trivial operator commutes with every jump, an exact strong symmetry cannot exist on the physical hardware. However, this symmetry is a property of the mathematical decomposition rather than the physical device. The unperturbed term \(\mathcal{L}_{0}\) is the analytically solvable, theoretically designed sector. This term is constructed either by restricting the coupling set to \(G\)-scalars or by isolating the \(G\)-symmetric coupling content within the sampler. The remaining terms form the covariant perturbation \(\eta\mathcal{L}_1\).The subsequent derivations rely on only two conditions. First, the kernel of \(\mathcal{L}_{0}^{\dag }\) restricted to the traceless subspace \(\mathbb{I}^{\perp }\) must consist entirely of the charge multiplets. These multiplets must be separated from the bulk spectrum by an \(O(1)\) spectral gap. Second, both \(\mathcal{L}_{0}\) and \(\mathcal{L}_{1}\) must remain \(G\)-covariant. A strong symmetry of \(\mathcal{L}_{0}\) is simply the natural mechanism that enforces this spectral structure, rather than an explicit assumption about the physical jump set.

\section{The slow-mode expansion and the warm start}\label{sec:setup}

Reversibility ensures that $\Lind^\dagger$ is self-adjoint under the KMS inner product Eq.\eqref{eq:kms}. I diagonalize this adjoint generator by constructing a KMS-orthonormal basis of eigen-observables $\{\ell_k\}_{k\ge0}$. These operators satisfy $\Lind^\dagger\ell_k=-g_k\ell_k$, where the relaxation rates are real
nonnegative:
\[
0 = g_0 < g_1 \le g_2 \le g_3 \le \cdots
\]
The stationarity mode corresponds to $\ell_0=\id$. The KMS detailed-balance identity takes the form
$\Lind\circ\Gamma_\sig=\Gamma_\sig\circ\Lind^\dagger$, with the channel defined as $\Gamma_\sig(X):=\sig^{1/2}X\sig^{1/2}$. This relation directly yields the Schr\"odinger-picture eigenmodes. The states $r_k:=\Gamma_\sig(\ell_k)$ satisfy
$\Lind r_k=-g_kr_k$, with the stationary state fixed at $r_0=\sig$. Both pictures share the exact same spectrum and differ only in their respective 
eigen-operators. 

Any deviation from stationarity expands in this basis as
\begin{equation}
\rho_t-\sig=\sum_{k\ge1}a_k(\rho_0)\,e^{-g_k t}\,r_k,
\qquad a_k(\rho_0)=\Tr\!\big[\ell_k^\dagger(\rho_0-\sig)\big]. 
\label{eq:expand}
\end{equation}
The expansion coefficients follow from the bi-orthogonality relation
$\Tr[\ell_j^\dagger r_k]=\inner{\ell_j}{\ell_k}_\sig=\delta_{jk}$. Each coefficient $a_k$ represents the
deviation of the observable $\ell_k$ from its thermal expectation value. In this framework, the $\chi^2$-divergence evaluates directly to 
$\chi^2(\rho_t\|\sig)=\sum_{k\ge1}|a_k(\rho_0)|^2e^{-2g_kt}$.

Three operational primitives can mitigate a slow mode bottleneck. First, one can remove the mode from the generator by introducing damping jumps or a non-reversible current. Second, one can monitor its decay and dynamically stop the process via self-certified halting. Third, one can project out the mode entirely within the initial state. This third approach requires no modification of the sampler, altering only the choice of the input state. It is therefore the most efficient strategy when re-engineering the dissipator is impractical or impossible. This mechanism forms the focus of the present section. I write the overlap with the slowest mode explicitly as:
 \begin{equation}
a_1(\rho_0)=\Tr\!\big[\ell_1^\dagger(\rho_0-\sig)\big].
\label{eq:a1}
\end{equation}
 Exact cancellation corresponds to the condition \(a_1(\rho_0)=0\) \cite{LuRaz,KlichRaz}. When the slowest rate is degenerate, meaning the slowest modes form a multiplet \(\{\ell_{1,i}\}_i\) under a non-Abelian symmetry, $a_{1}$ denotes the full vector of overlaps. The cancellation condition then requires the vanishing of the entire projection of $\rho_0-\sig$ onto this slow eigenspace, as eliminating a single coefficient is insufficient. This work relies strictly on this spectral definition of the effect. Vanishing slow-mode overlap guarantees a fundamentally faster asymptotic relaxation rate. Alternative trajectory-crossing formulations compare finite-time prefactors, where one input overtakes another in trace distance. While the spectral condition implies eventual overtaking of any generic input, it is not equivalent to a crossing at a pre-specified time.

I analyze the precise physical consequence of the cancellation condition \(a_1(\rho_0)=0\). The coefficient \(a_{1}\) quantifies the projection of the initial state deviation along the slowest relaxation mode. If this value is non-zero, this component exhibits the longest lifetime and dominates the late-time thermalization behavior. Conversely, if the coefficient is exactly zero, the slow mode is never populated. The system then relaxes according to the next, strictly faster rate \(g_{2}\). The remainder of this work develops methods to systematically achieve and exploit this vanishing condition.

\section{The rate-jump versus prefactor dichotomy}\label{sec:dichotomy}

Recall that the mixing time is evaluated in trace norm as $t_{\mathrm{mix}}(\rho_0;\epsilon)=\inf\{t:\tracenorm{\rho_t-\sig}\le\epsilon\}$. A warm start is defined as far if the initial divergence satisfies $\chi^2(\rho_0\|\sig)=O(1)$.

\begin{theorem}[Approximate cancellation scales prefactor; exact cancellation changes the rate]
\label{prop:dichotomy}
Under the ordered relaxation spectrum, the asymptotic limit  $\epsilon\to0$ yields two distinct regimes. If $a_1(\rho_0)\neq0$,
the mixing time satisfies $t_{\mathrm{mix}}(\rho_0;\epsilon)=g_1^{-1}\big[\log(|a_1|/\epsilon)+O(1)\big]$. Decreasing
$|a_1|$ within the interval $(0,O(1)]$ lowers the additive constant $g_1^{-1}\log|a_1|$ but leaves the leading rate
$g_1$ unchanged. An approximate warm start therefore constitutes a prefactor effect, saving at most
$g_1^{-1}\log(1/|a_1|)$ in time. 

Conversely, if $a_1(\rho_0)=0$, the slowest surviving relaxation mode is
governed by $g_2$. The mixing time becomes $t_{\mathrm{mix}}(\rho_0;\epsilon)=g_2^{-1}\big[\log(|a_2|/\epsilon)+O(1)\big]$, including a rate jump $g_1\!\to\!g_2$. The asymptotic ratio of mixing times satisfies
$t^{\mathrm{strong}}_{\mathrm{mix}}/t^{\mathrm{generic}}_{\mathrm{mix}}\to g_1/g_2\le1$. This
second advantage scales by a constant factor relative to the first unless $g_2\gg g_1$. Consequently, the rate jump provides a significant acceleration precisely when the slow mode is spectrally isolated. The $O(1)$ terms absorb all norm-conversion constants independent of 
$\epsilon$- and $\rho_0$. 
\end{theorem}

\begin{proof}
By the bi-orthogonality relation $\Tr[\ell_j^\dagger r_k]=\delta_{jk}$ and Eq.\eqref{eq:expand}, the divergence expands as 
$\chi^2(\rho_t\|\sig)=\sum_{k\ge1}|a_k|^2 e^{-2g_k t}$. Let $m=\min\{k:a_k\neq0\}$ denote the index
of the slowest populated mode. Factoring out this dominant term yields
\[
\chi^2(\rho_t\|\sig)=|a_m|^2 e^{-2g_m t}\Big(1+\textstyle\sum_{k>m}|a_k/a_m|^2 e^{-2(g_k-g_m)t}\Big)
=|a_m|^2 e^{-2g_m t}\big(1+o(1)\big),
\]
since every exponent $g_k-g_m>0$ for $k>m$. The parenthetical sum therefore approaches unity as $t\to\infty$.

This asymptotic scaling sandwiches the trace distance between two exponentials sharing the same rate. Pairing the state deviation with the slow observable yields the strict lower bound:
$\tracenorm{\rho_t-\sig}\ge|\Tr[\ell_m^\dagger(\rho_t-\sig)]|/\|\ell_m\|_\infty
=(|a_m|/\|\ell_m\|_\infty)\,e^{-g_mt}$. The corresponding upper bound follows from the $\chi^2$-divergence inequality 
$\tracenorm{\rho_t-\sig}\le|a_m|e^{-g_mt}(1+o(1))$. Solving both bounds at the target error level $\epsilon$ establishes 
$t_{\mathrm{mix}}(\rho_0;\epsilon)=g_m^{-1}\big[\log(|a_m|/\epsilon)+O(1)\big]$. The $O(1)$ residual is bounded
between $0$ and $\log\|\ell_m\|_\infty$, independent of $\epsilon$ and of $\rho_0$.
If $a_1\neq0$, the index is $m=1$, which recovers the first scaling law. The dependence on $|a_1|$ enters solely through  the additive
term $g_1^{-1}\log|a_1|$. Varying $|a_1|$ over $(0, O(1)]$ leaves the rate fixed at $g_1$ and alters only this constant, acting as a prefactor reduction. If instead $a_1=0$, the slowest populated mode is generically
$m=2$, yielding the rate jump $g_1\to g_2$. Evaluating the ratio of these two mixing times yields:
\[
\frac{t^{\mathrm{strong}}_{\mathrm{mix}}}{t^{\mathrm{generic}}_{\mathrm{mix}}}
=\frac{g_2^{-1}\log(|a_2|/\epsilon)}{g_1^{-1}\log(|a_1|/\epsilon)}\xrightarrow[\epsilon\to0]{}\frac{g_1}{g_2}\le1,
\]
where the logarithmic terms agree to leading order in $\epsilon$. This ratio remains bounded below by $1/2$ unless
$g_2\ge 2g_1$. This demonstrates that the rate jump provides a substantial speedup precisely under special isolation, where $g_2\gg g_1$.
\end{proof}

We analyze this behavior geometrically as a shift versus a tilt on a log-distance-versus-time plot. An approximate warm start merely shifts the convergence curve downward, maintaining the original final slope $g_1$ from a lower starting coordinate. Exact cancellation tilts the trajectory, steepening the final asymptotic slope to $g_2$. While a shift saves a fixed constant, a tilt modifies the asymptotic scaling. The tilt is considerably more powerful, but its physical impact requires sufficient spectral separation between $g_2$ and $g_1$ to manifest.

\begin{remark}[Exactness criterion]\label{rem:exact}
A non-zero overlap coefficient $a_1$ ceases to dominate this late-time mixing behavior once the relaxation timeline of the first mode falls below that of the second mode. This transition occurs when the initial overlap satisfies the scaling threshold \(\vert{}a_1\vert{} \lesssim a_1^\star := \epsilon^{1-g_1/g_2} \vert{}a_2\vert{}^{g_1/g_2}\). Under strong spectral separation where \(g_1 \ll g_2\), this stability threshold simplifies to \(a_1^\star \approx \epsilon\). Capturing the true rate jump therefore requires suppressing the initial slow-mode overlap down to the target precision tolerance. This scaling requirement explains why generic warm starts yield only a prefactor improvement, and underscores why achieving a genuine rate jump demands exact a priori knowledge of the slow observable \(\ell _{1}\).
\end{remark}

For a generic reversible sampler, exact cancellation requires explicit knowledge of \(\ell _{1}\). This operator is the slowest eigenvector of a superoperator acting on a space of dimension \(d^2 = 4^n\). Computing this eigenvector is generically an exponential task, making it as difficult as the sampling problem itself. The scaling laws proved above establish exactly when an exact initialization strategy becomes viable. Theorem~\ref{prop:dichotomy}  and Remark~\ref{rem:exact} together fix the precise domain of validity for the strategy. Specifically, two conditions must be satisfied simultaneously. First, the slow mode must be spectrally isolated such that \(g_2 \gg g_1\). Second, the operator \(\ell _{1}\) must be accessible at a low computational cost. The central observation of this section is that both conditions are simultaneously satisfied under the single structural constraint of a weakly broken symmetry.

\section{Symmetry gives the bottleneck eigenvector for free}\label{sec:symsec}

\begin{definition}[Weakly broken strong symmetry]
I define the total system dynamics via the perturbed Liouvillian $\Lind=\Lind_0+\eta\,\Lind_1$. Here, the unperturbed channel $\Lind_0$ acts as a reversible sampler possessing
an exact \emph{strong symmetry}. This symmetry is generated by a Hermitian charge $Q$ in the kernel of the adjoint operator, $\Lind_0^\dagger \Qch=0$, while an $O(1)$ spectral gap isolates the remaining relaxation modes. The term $\eta\,\Lind_1 (\eta \ll 1)$ introduces a weak, reversible perturbation that explicitly breaks the
conservation of $Q$ but maintains the identity as a steady state, satisfying   $\Lind_1^\dagger\mathbb{I}=0$.

\end{definition}

\begin{theorem}[The slow mode is the charge]
\label{prop:symmetry}
For small $\eta$, the perturbed Liouvillian $\Lind$ has an isolated slow mode of rate
\[
g_1=\eta\,q+O(\eta^2),\qquad
q=-\frac{\inner{\Qch}{\Lind_1^\dagger \Qch}_\sig}{\normsig{\Qch}^2}\ge0.
\]
The bulk rates remain $g_{k\ge2}=O(1)$. This separation of timescale establishes a wide spectral ratio $g_2/g_1=O(1/\eta)$. The corresponding slow left eigenvector tracks the charge via $\ell_1=\Qch+O(\eta)$. To leading order, the slowest mixing mode is the symmetry charge.
Consequently, the exact-projection condition Eq.\eqref{eq:a1} reduces to the leading-order expectation matching requirement:
\begin{equation}
{\ \langle \Qch\rangle_{\rho_0}=\langle \Qch\rangle_\sig\ }.
\label{eq:warmcond}
\end{equation}
This match leaves a tiny residual overlap bounded by $a_1(\rho_0)=O(\eta)$.
\end{theorem}

\begin{proof}
Both $\Lind_0^\dagger$ and $\Lind_1^\dagger$ are self-adjoint in 
$\inner{\cdot}{\cdot}_\sig$ due to reversibility. The condition $\Lind^\dagger\mathbb{I}=0$ holds for every $\eta$
because $\sig$ stays stationary. The identity $\mathbb{I}$ is an exact eigenvector at $0$ throughout. 
I restrict to its orthogonal complement $\mathbb{I}^\perp$, the $\sig$-mean-zero sector. 
On $\mathbb{I}^\perp$ the strong symmetry makes $0$ a simple eigenvalue of $\Lind_0^\dagger$.  The centered eigenvector $Q$ satisfies $\inner{Q}{\mathbb{I}}_\sig=0$. A gap $\Delta_0=O(1)$ isolates 0 from the rest of the spectrum. The eigenvalue is isolated. Kato analytic perturbation theory applies to $\Lind^\dagger=\Lind_0^\dagger+\eta\,\Lind_1^\dagger$. 

For small $\eta$ there is an analytic
eigenvalue branch $-g_1(\eta)$ and eigenvector $\ell_1(\eta)$. The first-order data read
\[
g_1(\eta)=-\eta\,
\frac{\inner{Q}{\Lind_1^\dagger Q}_\sig}{\normsig{Q}^2}+O(\eta^2)=\eta q+O(\eta^2),
\qquad
\ell_1(\eta)=Q-\eta\,R_0\,\Lind_1^\dagger Q+O(\eta^2).
\]
The operator $R_0$ is the reduced resolvent of $\Lind_0^\dagger$ on $\{\mathbb{I},Q\}^\perp$. Both operators
are real. The eigenvalue simple. The branch is real for small $\eta$. Dissipativity $\Lind^\dagger\le0$ forces $\eta\inner{Q}{\Lind_1^\dagger Q}_\sig\le0$. This condition yields $q\ge0$. 

The remaining eigenvalues shift by $O(\eta)$ from their $O(1)$ values. I find  $g_{k\ge2}=O(1)$ and
$g_2/g_1=O(1/\eta)$. This gives the claimed rate and the eigenvector. 

For the warm-start condition, substitute $\ell_1=Q+O(\eta)$ into the initial overlap Eq.\eqref{eq:a1}. This step yields 
$a_1(\rho_0)=\Tr[Q^\dagger(\rho_0-\sig)]+O(\eta)=\big(\langle Q\rangle_{\rho_0}-\langle Q\rangle_\sig\big)+O(\eta)$.
To leading order this vanishes exactly when Eq.\eqref{eq:warmcond} holds.
\end{proof}

The identification of the charge as the slow mode rests on a simple physical mechanism.
A strong symmetry makes $Q$ exactly conserved. 
It creates an infinitely slow zero mode within the unperturbed dynamics. Introducing weak symmetry breaking of strength $\eta$
lifts this zero mode to a small decay rate $g_1=O(\eta)$. Being far below the $O(1)$ bulk, this perturbation alters the underlying eigenvector only slightly. The slowest-mixing observable is still $Q$ to leading
order. This almost-conserved quantity represents the exact feature that the dynamics takes longest to equilibrate. 
The presence of the symmetry identifies the slow left eigenvector $\ell_1$ directly without requiring any explicit calculation.

The expectation matching condition Eq.~\eqref{eq:warmcond} provides  the entire prescription for accelerating equilibration.
Preparing an initial state whose charge expectation matches the thermal value projects out this slowest mode completely. 
The sampler subsequently mixes at the fast bulk rate. 
No spectral computation is required. 
The identity of the symmetry is the sole necessary input. 
The symmetry replaces the need for spectral estimation by identifying the slow observable strictly from the algebraic structure. 
Projecting out this slow mode becomes a straightforward condition imposed on the input state rather than a complex quantity to calculate. The subsequent sections carry this identical mechanism to the non-Abelian setting. there is a single charge expectation constraint generalizes to a group invariance condition. The achievable speedups form a lattice structure.

\begin{corollary}[Speedup from a symmetry warm start]\label{cor:warmalg}
Given a weakly broken symmetry with charge $Q$, the warm start expectation matching condition $\langle Q\rangle_{\rho_0}=\langle Q\rangle_\sig$ projects out the slow mode. 
This preparation requires no spectral estimation. 
The mixing time drops to $t_{\mathrm{mix}}=O\! \big(g_2^{-1}\log(1/\epsilon)\big)=O(\log(1/\epsilon))$
when the bulk rate satisfies $g_2=\Theta(1)$. 
This reduction yields an $\Omega(1/\eta)=\Omega(g_2/g_1)$ speedup over a generic initialization. 
The initialization incurs on $O(1)$ preparation overhead. 
For a non-Abelian group $G$, the equivalent requirement becomes full $G$-invariance, expressed as $\mathsf{A}_G(\rho_0)=0$. 
The subgroup-invariant warm starts subsequently yield the structured speedup cascade, to be described below in Theorem~\ref{prop:nonabelian}.

\end{corollary}

\section{The non-Abelian asymmetry cascade}\label{sec:nonabelian}

The single charge theorem above projects out a single conserved quantity. A compact non-Abelian symmetry group $G$ expands the slow manifold dimensions. The correct initialization target shifts from a single expectation value to an explicit 'asymmetry' measure.  Let $g\mapsto U_g$
define a unitary representation of $G$ acting on the Hilbert space $\Hilb$. The conjugation action operates via 
$\mathcal{U}_g(X)=U_g X U_g^\dagger$. The Haar-average $\mathcal{G}=\int_G\mathrm{d}g\,\mathcal{U}_g$ forms the orthogonal projection onto the $G$-invariant sub-algebra. The 'asymmetry' of a state is
the relative entropy calculated with respect to its own Haar-average image
\[
\mathsf{A}_G(\rho)=S\big(\rho\,\big\|\,\mathcal{G}(\rho)\big)=S(\mathcal{G}\rho)-S(\rho)\ \ge 0. 
\]
This informational asymmetry measure vanishes exactly when the state $\rho$ exhibits full $G$-invariance~\cite{MarvianSpekkens,AresMurcianoCalabrese}.

\begin{definition}[Weakly broken strong $G$-symmetry]\label{def:nonabelian}
Let $\Lind=\Lind_0+\eta\Lind_1$ define the total perturbed generator. The unperturbed component $\Lind_0$ is a reversible operator satisfying KMS--detailed-balanced with
respect to a $G$-invariant state $\sig$ ($\mathcal{U}_g\sig=\sig$). The unperturbed operator $\Lind_0$ carry a
'strong' $G$-symmetry. Every constituent jump operator is invariant under the $G$ action. This implies $[\mathcal{U}_g,\Lind_0]=0$. Every separate $G$-tensor charge operator is conserved. 
Assume $\Lind_0$ is $G$-ergodic. Evaluate the system within $\mathbb{I}^\perp$. The kernel consists precisely of the linear span of a finite family of charges $\{Q_{\lambda,i}\}$. These charges transform under nontrivial irreducible representations labelled by $\lambda\in\Lambda$ with internal components running over $i=1,\dots,d_\lambda$. A substantial spectral gap scaled at $O(1)$ isolates these kernel modes from the remaining relaxation spectrum. The perturbation $\eta\Lind_1$ is a reversible operator preserving the stationarity of $\sig$. It is 
$G$-covariant, satisfying $[\mathcal{U}_g,\Lind_1]=0$. It breaks the conservation laws while preserving the underlying group symmetry properties.
\end{definition}

The Abelian Definition of Section~\ref{sec:symsec} is the case $G=\mathrm{U}(1)$,
$\Lambda=\{1\}$, $Q_1=Q$.

\begin{lemma}[Covariant generators are irreducible-representation-block-diagonal]\label{lem:cov}
Let $\Lind$ be a reversible, $G$-covariant generator. Under the conjugation action $\mathcal{U}_g$, 
the operator space decomposes into $G$-isotypic components,
$\mathcal{B}(\Hilb)=\bigoplus_{\lambda}\mathcal{B}_\lambda$. The adjoint operator $\Lind^\dagger$ preserves each component $\mathcal{B}_\lambda$. On the multiplicity space of a fixed irreducible representation, $\Lind^\dagger$ commutes with the irreducible representation action. Consequently, every
eigen-operator of $\Lind$ belongs to a definite irreducible representation. All operators within the same irreducible representation multiplet decay at
the same rate $r_\lambda$. Asymmetric observables belonging to the nontrivial irreducible representation operators never
mix with the symmetric observables belonging to the trivial-irreducible-representation sector, which forms \ the commutant of $\{U_g\}$).
\end{lemma}

\begin{proof}
Covariance $[\mathcal{U}_g,\Lind]=0$ establishes that $\Lind^\dagger$ acts as an intertwiner of the representation
$\mathcal{U}$ carried by $\mathcal{B}(\Hilb)$. 
By Schur's lemma, an intertwiner is block-diagonal across inequivalent isotypic components. 
On the multiplicity space of a fixed irreducible representation, the intertwiner commutes with the irreducible action. 
The intertwiner acts as a scalar on each irreducible copy. 
It also acts as an arbitrary operator across copies of the {same} irreducible representation. 
Reversibility ensures $\Lind^\dagger$ is self-adjoint under the KMS iner product Eq.\eqref{eq:kms}.
As $\sig$ is $G$-invariant, the isotypic decomposition is $\inner{\cdot}{\cdot}_\sig$-orthogonal. 
The spectral decomposition therefore refines the isotypic
one. 
This places each eigenoperator inside a single $\mathcal{B}_\lambda$.
\end{proof}

Lemma~\ref{lem:cov} is exact. It avoids perturbation theory entirely. It forms the structural backbone of the
cascade. The relation holds for $\Lind=\Lind_0+\eta\Lind_1$ at every value of $\eta$. Slow modes carry a definite
irreducible representation across the entire spectrum. The perturbation merely determines their numerical rates. This block-diagonalization of a
covariant detailed-balanced spectrum constitutes the standard Schur decomposition. I use it for the mechanism that converts initialization into
representation theory via the asymmetry target and the cascade of Section~\ref{sec:nonabelian}.

\begin{lemma}[Schur splitting of the slow rates]\label{lem:schur}
Under Definition~\ref{def:nonabelian}, the degenerate kernel
$\mathcal{Q}=\bigoplus_{\lambda\in\Lambda}V_\lambda$ of $\Lind_0^\dagger$ on $\mathbb{I}^\perp$ lifts to
first order in $\eta$ to irreducible representation-labelled rates
\[
g_\lambda=\eta\,q_\lambda+O(\eta^2),\qquad
q_\lambda=-\frac{\inner{Q_{\lambda,i}}{\Lind_1^\dagger Q_{\lambda,i}}_\sig}{\normsig{Q_{\lambda,i}}^2}\ge 0. 
\]
These rates are independent of $i\in\{1,\dots,d_\lambda\}$ within an irreducible representation. The bulk stays $O(1)$.
\end{lemma}

\begin{proof}
The kernel is degenerate. First-order shifts come from Kato degenerate perturbation theory. The
$O(\eta)$ rates are the eigenvalues of $\Gamma^{-1}M$. Here, 
$\Gamma_{(\lambda i),(\mu j)}=\inner{Q_{\lambda,i}}{Q_{\mu,j}}_\sig$ is the KMS Gram matrix of the kernel, 
and $M_{(\lambda i),(\mu j)}=-\inner{Q_{\lambda,i}}{\Lind_1^\dagger Q_{\mu,j}}_\sig$ is the breaking form restricted to $\mathcal{Q}$. 
Both $\Gamma$ and $M$ are intertwiners of the $G$-action
on $\mathcal{Q}$. 
The $\Gamma$ intertwines because $\sig$ is $G$-invariant. The $M$ intertwines because $\Lind_1$ is $G$-covariant
(Lemma~\ref{lem:cov}). 
By Schur's lemma,  each matrix is block-diagonal across distinct irreducible representations. 
Each matrix is a scalar on each $V_\lambda$, say $\Gamma|_{V_\lambda}=\gamma_\lambda I$ and $M|_{V_\lambda}=\mu_\lambda I$. The product $\Gamma^{-1}M$ is scalar on $V_\lambda$. It has the single eigenvalue $g_\lambda=\mu_\lambda/\gamma_\lambda$.
This value matches the displayed Rayleigh quotient. 
It is identical for every $i$. 
The branch is real for small $\eta$. The adjoints $\Lind_0^\dagger,\Lind_1^\dagger$ are
self-adjoint. 
The kernel is isolated. Dissipativity $\Lind^\dagger\le0$ forces
$\mu_\lambda\ge0$. 
This condition implies $q_\lambda\ge0$. 
The non-kernel eigenvalues of $\Lind_0^\dagger$ are $O(1)$.
They move by only $O(\eta)$. 
The bulk stays $O(1)$. 
Each slow rate $g_\lambda=O(\eta)$ is spectrally isolated. 
This isolation defines the regime where the rate jump of Theorem~\ref{prop:dichotomy} functions.
\end{proof}

The first-order formula prompts the question of how far it extends: is the linearity in $\eta$ an
artifact of small-$\eta$ asymptotic limit, and do the constants survive as the system grows? Both
questions have sharp answers without of any perturbation theory.

\begin{proposition}[Non-perturbative sandwich for the slow rates]\label{prop:sandwich}
Dress observables by $X\mapsto\sig^{1/4}X\sig^{1/4}$. This mapping defines an isometry from the KMS space onto
the Hilbert--Schmidt space. It maps $-\Lind_i^\dagger$ to a positive semidefinite self-adjoint operator $\mathsf H_i$.
Here, $\mathsf H_i\,\sig^{1/2}=0$ and
$\spec(-\Lind_0-\eta\Lind_1)=\spec(\mathsf H_0+\eta\mathsf H_1)$. 
Let $P$ project onto $\ker\mathsf H_0\ominus\sig^{1/2}$ (the dressed charge multiplets).
Let $g_{\mathrm{bulk}}$ be the smallest nonzero eigenvalue of $\mathsf H_0$.
For each $\lambda\in\Lambda$, let $q_\lambda$ be
the Schur scalar of $P\mathsf H_1P$ on the $\lambda$-multiplet. 
This matches the Rayleigh quotient as in
Lemma~\ref{lem:schur}). 
Let $b_\lambda:=\|(\mathbb{I}-P)\mathsf H_1P_\lambda\|$. 
Then, for every $\eta>0$ satisfying $\eta q_\lambda<g_{\mathrm{bulk}}$,
\[
\eta\,q_\lambda-\frac{\eta^2\,b_\lambda^2}{g_{\mathrm{bulk}}-\eta q_\lambda}
\;\le\;g_\lambda(\eta)\;\le\;\eta\,q_\lambda .
\]
The upper bound holds exactly at all $\eta$. 
Moreover, if $q_\lambda>0$ for all $\lambda$, the full gap
satisfies the global linear bound $\gap(\Lind_0+\eta\Lind_1)\ge\eta\,\gap(\Lind_1)$ for all $\eta\ge0$,
and $\eta\mapsto\gap(\Lind_0+\eta\Lind_1)$ is concave. The linear regime is global, not asymptotic.
\end{proposition}

\begin{proof}
The isometry and positivity follow the standard vectorization of a KMS-reversible generator. 
The identity $\sig^{1/2}$ is the image of $\mathbb{I}$. 
Lemma~\ref{lem:cov} ensures $\mathsf H(\eta)=\mathsf H_0+\eta\mathsf H_1$
is block-diagonal over the $G$-isotypic components at every $\eta$. The argument proceeds entirely within the
$\lambda$-isotypic block. 
The labels are unambiguous. 
\emph{Upper bound.} 
The block contains the multiplet $\mathcal K_\lambda\subset\ker\mathsf H_0$. 
The quadratic form of $\mathsf H(\eta)$ equals $\eta q_\lambda$ identically on this subspace by Schur's lemma. 
Courant--Fischer using test space $\mathcal K_\lambda$ bounds the lowest $\dim\mathcal K_\lambda$ eigenvalues of the block by $\eta q_\lambda$. 
\emph{Lower bound.} 
Split the block via $P_\lambda$ and $\mathbb{I}-P_\lambda$. In this
decomposition,
\[
\mathsf H(\eta)=\begin{pmatrix}\eta q_\lambda\,\mathbb{I} & \eta B^\dagger\\[2pt] \eta B & C+\eta D\end{pmatrix},
\qquad C\succeq g_{\mathrm{bulk}},\quad D\succeq0,\quad\|B\|\le b_\lambda .
\]
The off-diagonal blocks of $\mathsf H_0$ vanish because $P_\lambda$ projects into its kernel. 
Any eigenvalue $\mu<g_{\mathrm{bulk}}$ of the block satisfies the Schur-complement equation $\mu\in\spec\big(\eta q_\lambda\mathbb{I}-\eta^2B^\dagger(C+\eta D-\mu)^{-1}B\big)$.
The operator bound $0\preceq B^\dagger(C+\eta D-\mu)^{-1}B\preceq b_\lambda^2/(g_{\mathrm{bulk}}-\mu)\,\mathbb{I}$ holds. 
Substituting $\mu\le\eta q_\lambda$ from the upper bound yields the inequality. 
\emph{Global bound.} 
The condition
$q_\lambda>0$ for all $\lambda$ implies $\ker\mathsf H_0\cap\ker\mathsf H_1=\mathrm{span}\{\sig^{1/2}\}$. 
On the orthogonal complement of this common kernel, the smallest eigenvalue of a nonnegative combination is a concave function of the coefficients. 
It is an infimum of affine functions (Lemma~22 of \cite{Onorati}).
Therefore, 
$\gap(\mathsf H_0+\eta\mathsf H_1)\ge\gap(\mathsf H_0)+\eta\,\gap(\mathsf H_1)=\eta\,\gap(\mathsf H_1)$.
The charges lie in the complement and are annihilated by $\mathsf H_0$, so $\gap(\mathsf H_0)=0$ in this convention.
Concavity combined with vanishing at $\eta=0$ makes
$\gap(\eta)/\eta$ non-increasing. The $\eta$-dependence of the bound cannot be improved.
\end{proof}

\begin{remark}[Uniformity in system size]\label{rem:size}
Proposition~\ref{prop:sandwich} reduces the system-size independence of the breaking strength to three static constants: \(q_{\lambda }\), \(b_{\lambda }\), and \(g_{\mathrm{bulk}}\). 
This condition is equivalent to holding \(g_\lambda=\Theta(\eta)\) uniformly in \(n\). 
No dynamical input remains. 
The condition \(g_{\mathrm{bulk}}=\Omega(1)\) states that $\Lind_0$ mixes rapidly within symmetry sectors. 
This defines the target regime of the construction.
For a quasi-local covariant bath and collective charge multiplets, \(q_{\lambda }\) is a ratio of two extensive quadratic forms and \(b_{\lambda }\) collects local contributions against a normalized collective mode. 
Section~\ref{sec:example} contains a worked model where all three constants are \(O(1)\). 
The exact linearity of every slow rate in \(\eta \) observed there reflects the sandwich collapsing onto its upper edge. 
Establishing \(q_\lambda=\Theta(1)\) and \(b_\lambda=O(1)\) for lattice model families requires a purely static form-factor computation. 
This computation represents the most consequential open problem posed by this framework.
\end{remark}

\begin{theorem}[Non-Abelian warm start: the asymmetry cascade]\label{prop:nonabelian}
Order the distinct slow rates $g_{\lambda_1}<g_{\lambda_2}<\dots<O(1)$. Then, as $\epsilon\to0$:
\emph{(i) Full jump.} 
Let $\rho_0$ be $G$-invariant, meaning $\mathsf{A}_G(\rho_0)=0$.
Every slow overlap vanishes. 
The mixing time depends entirely on the symmetric bulk.
This condition creates a full rate jump from
$g_{\lambda_1}$ to $O(1)$, yielding $t_{\mathrm{mix}}=O\!\big(\log(1/\epsilon)\big)$. 
\emph{(ii) Quantitative.} 
In general, surviving slow overlaps satisfy
$\sum_{\lambda,i}|a_{\lambda,i}(\rho_0)|^2\le C\,\mathsf{A}_G(\rho_0)$ with
$C=2\sum_{\lambda,i}\|Q_{\lambda,i}\|_\infty^2$ via quantum Pinsker.
A small asymmetry yields a small residual and a proportionately  large
prefactor saving. 
\emph{(iii) Cascade.} 
Let $\rho_0$ be invariant only under a subgroup
$H\le G$ ($\mathcal{G}_H\rho_0=\rho_0$).
Irreducible representations lacking a  $H$-fixed vector are projected out.
The surviving rate matches the slowest $g_\lambda$ among the remaining
$H$-active irreducible representations.
The achievable jumps thus form a lattice indexed by the subgroup lattice of $G$.
\end{theorem}

\begin{proof}
Lemma~\ref{lem:schur} identifies the slow eigen-observables as the $Q_{\lambda,i}$ up to $O(\eta)$.
This identification implies 
$a_{\lambda,i}(\rho_0)=\Tr[Q_{\lambda,i}^\dagger(\rho_0-\sig)]+O(\eta)$.
To leading order, slow overlaps equal the deviations of the charge-multiplet expectations from thermal equilibrium.
The $Q_{\lambda,i}$
carry nontrivial irreducible representations, forcing $\mathcal{G}(Q_{\lambda,i})=0$.
The averaging is self-adjoint and trace-preserving with $\mathcal{G}\sig=\sig$.
These properties imply
$\Tr[Q_{\lambda,i}^\dagger\,\mathcal{G}\rho_0]=\Tr[(\mathcal{G}Q_{\lambda,i})^\dagger\rho_0]=0$ and
$\Tr[Q_{\lambda,i}^\dagger\sig]=0$.
Therefore, 
$a_{\lambda,i}=\Tr[Q_{\lambda,i}^\dagger(\rho_0-\mathcal{G}\rho_0)]+O(\eta)$.
This bound yields
$|a_{\lambda,i}|\le\|Q_{\lambda,i}\|_\infty\tracenorm{\rho_0-\mathcal{G}\rho_0}+O(\eta)$. 
Statement (i) corresponds to $\mathcal{G}\rho_0=\rho_0$. Statement (ii) follows from the quantum Pinsker inequality
$\tfrac12\tracenorm{\rho_0-\mathcal{G}\rho_0}^2\le \mathsf{A}_G(\rho_0)$, yielding the constatnt
$C=2\sum_{\lambda,i}\|Q_{\lambda,i}\|_\infty^2$. 
For statement (iii), the condition $\mathcal{G}_H\rho_0=\rho_0$ makes
$\rho_0$ $\sig$-orthogonal to every irreducible representation lacking an $H$-fixed vector because the $H$-averaging of
such a multiplet vanishes.
The remaining overlaps persists. 
Applying the Abelian dichotomy
(Theorem~\ref{prop:dichotomy}) to the slowest surviving rate determines the mixing time.
\end{proof}

\begin{corollary}[First moments are not enough]\label{cor:firstmoments}
In the Abelian case, the entire warm start is the single linear constraint Eq.\eqref{eq:warmcond}: match one charge expectation. 
Porting this recipe to a non-Abelian $G$ by preparing $\rho_0$ with all first moments of the conserved charges thermal,
$\langle Q_{\lambda_1,i}\rangle_{\rho_0}=\langle Q_{\lambda_1,i}\rangle_\sig$ fails to produce the rate jump.
Matching the first moments nulls only the overlaps with the defining (adjoint)
multiplet.
Every higher slow multipole (quadrupole and beyond) in $\Lambda$ retains its own $O(\eta)$ rate and becomes the new bottleneck. 
This matches exactly one step up the cascade of
Theorem~\ref{prop:nonabelian}(iii).
It does not produce the full jump. 
The full jump requires $\mathsf{A}_G(\rho_0)=0$.
This condition demand genuine $G$-invariance. 
This condition is nonlinear in $\rho_0$. 
An $O(1)$ cost achieves this condition via the Haar-averagin
$\mathcal{G}(\rho)$ of any state, the maximally mixed state, or by any $G$-symmetric product state. 
Concretely, consider the $SU(2)$ sampler I will study in depth in Section~\ref{sec:example}.
An input with thermal first moment $\langle\vec S_{\mathrm{tot}}\rangle$ but non-thermal quadrupole coherences
$\langle S^aS^b\rangle$ project out the spin-$1$ triplet. 
It then mixes at the slower spin-$2$ rate rather than the
bulk rate.
\end{corollary}

\begin{remark}[Measurability and the logical directions]\label{rem:asym}
The single-charge case recovers the logical-sector initialization of
\cite{BergamaschiGheissariLiu}.
The non-Abelian multiplet, its irreducible-representation-indexed cascade, and the asymmetry target constitute the additions of this framework. 
The target is measurable, $\mathsf{A}_G$ and its
$H$-residuals follow from charged moments through the randomized-measurement primitives.
Self-certifying samplers already use these primitives.
One estimates initialization quality directly from the
prepared state.
Consider a representation carrying multiplicities:
$U_g=\bigoplus_\lambda D_\lambda(g)\otimes\mathbb{I}_{m_\lambda}$. The commutant carries logical operators on the multiplicity spaces.
This structure defines a noiseless subsystem.
The input matches this thermal data separately. 
The asymmetry clears the gauge directions.
The logical directions are matched independently. 
This mechanism represents the open-system face of a quantum code.
I will further discuss this issue in the outlook section.
\end{remark}

\section{A concrete example: a genuine \texorpdfstring{$SU(2)$}{SU(2)} Davies sampler}\label{sec:example}

I now evaluate the cascade on the simplest non-Abelian example: a concrete, fully reversible $SU(2)$ Davies sampler. 
I check every structural hypothesis of Definition~\ref{def:nonabelian} via  numerical diagonalization rather than assumed.

\paragraph{Model.}
Take $N=4$ spin-$\tfrac12$ sites with the non-uniform Heisenberg chain
$H=\sum_{i=1}^{3}J_i\,\vec S_i\!\cdot\!\vec S_{i+1}$. The couplings are $(J_1,J_2,J_3)=(1.0,0.7,1.3)$, at $\beta=1.2$. 
This Hamiltonian possesses $SU(2)$-symmetry, $[H,\vec S_{\mathrm{tot}}]=0$.
The stationary state $\sig$ is rotation invariant.
The non-uniform couplings split the spin multiplets in energy.
Consequently, scalar couplings drive genuine relaxation within each total-spin sector.
I construct two Davies generators with
the KMS rate $\gamma(\omega)=(1+e^{\beta\omega})^{-1}$.
The coherent term is omitted, leaving a purely
dissipative reversible generator.
\begin{itemize}[leftmargin=1.4em,itemsep=1pt]
\item $\Lind_0$ is built from the $SU(2)$-{scalar} couplings $\{\vec S_i\!\cdot\!\vec S_j\}_{i<j}$. 
These jumps commute with $\vec S_{\mathrm{tot}}$. Therefore, $\Lind_0$ carries a 'strong' $SU(2)$ symmetry and conserves total spin.
\item $\Lind_1$ is built from the 'vector' couplings $\{S^a_i\}_{a=x,y,z;\,i}$. This rotation-closed set makes $\Lind_1$ is $SU(2)$-covariant but relaxes total spin.
\end{itemize}
The full generator is $\Lind=\Lind_0+\eta\Lind_1$.

\paragraph{Verified structure.}
Full diagonalization of the $256$-dimensional Liouvillian verifies each hypothesis to machine precision.
The thermal state $\sig$ is stationary for both parts, satisfying ($\|\Lind_{0,1}(\sig)\|\sim10^{-11}$).
Both are trace-preserving and GNS-reversible, yielding a real spectrum to within $10^{-16}$).
The unperturbed $\Lind_0$ conserves the total spin, verified by
$|\Tr(\vec S_{\mathrm{tot}}\,\Lind_0(\rho))|\sim5\times10^{-16}$, while
$\Lind_1$ breaks this conservation.
Both terms are $SU(2)$-covariant, with ($\|[\mathcal{U}_g,\Lind_{0,1}]\|\sim10^{-14}$).
The eigen-operators of $\Lind$ carry definite spin as demanded by  Lemma~\ref{lem:cov}.
Applying the adjoint Casimir $\mathcal{C}(X)=\sum_a[S^a_{\mathrm{tot}},[S^a_{\mathrm{tot}},X]]=j(j+1)X$ to each slow mode returns integers $j(j+1)\in\{2,0,6\}$ to within $10^{-12}$.

\paragraph{The slow spectrum and the cascade.}
At $\eta=0.15$, the slow modes are exactly the $SU(2)$ multipoles of the total spin:
\[
\underbrace{g\approx0.101}_{\text{spin-1 dipole }(\times3)}\ <\
\underbrace{0.137}_{\text{spin-0 scalar}}\ <\
\underbrace{0.191}_{\text{spin-2 quadrupole }(\times5)}\ \ll\ O(1)\ \text{bulk}.
\]
Each rate scales 'linearly' in $\eta$. Halving $\eta$ halves every slow rate to a ratio of $2.00$.
This confirms the $g_\lambda=\eta q_\lambda$ scaling of Lemma~\ref{lem:schur}. 
Figure~\ref{fig:cascadeB}
illustrates the resulting cascade via the trace distance to $\sig$. 
A generic input is limited by the
spin-1 dipole, which forms the globally slowest mode. 
An $SU(2)$-invariant input where $\mathsf{A}_G=0$ is prepared by
the exact Haar-averaging. It lacks dipole or quadrupole content and mixes faster, though it does not reach the bulk rate.

\paragraph{What the example exposes.}
The residual bottleneck of the symmetric input is the spin-$0$ 'scalar' mode at $0.137$. 
This direction is an $SU(2)$-invariant slow mode that the averaging cannot remove because it belongs to the trivial irreducible representation.
It represents the conserved population imbalance between total-spin sectors.
This corresponds to the logical data of the
strong symmetry rather than gauge data. 
This numerical test quantifies the multiplicity and commutant subtlety described in 
Remark~\ref{rem:asym}.
The $\mathsf{A}_G$ targets the nontrivial-irreducible-representation multiplets, which are the
dipole and quadrupole in this model.
Minimizing it yields a real but partial speedup over $0.101\to0.137$. 
The full jump to the bulk requires additionally matching the conserved scalar sector data separately, leaving only the $O(1)$ bulk modes. 
The sampler realizes both components of the theory simultaneously: the asymmetry cascade and the separate matching of logical data. 

\begin{figure}[t]
\centering
\includegraphics[width=0.74\linewidth]{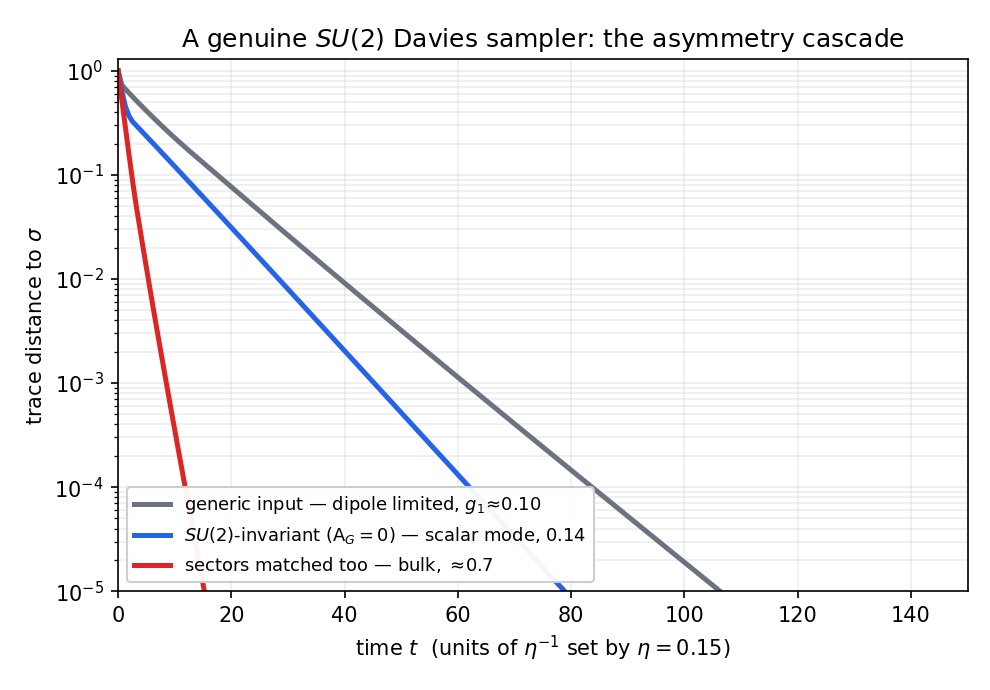}
\caption{The asymmetry cascade on the genuine $SU(2)$ Davies sampler of Section~\ref{sec:example}
($N=4$ Heisenberg, $\beta=1.2$, $\eta=0.15$; trace distance to $\sig$, all rates verified against the
diagonalized generator). \textbf{Grey:} a generic input is limited by the spin-$1$ dipole, the
slowest mode, $g_1\approx0.10$. \textbf{Blue:} an $SU(2)$-invariant input ($\mathsf{A}_G=0$) has
nulled the dipole and quadrupole and mixes faster, but is limited by a residual $SU(2)$-\emph{invariant}
slow mode---the conserved sector population, rate $0.14$ (Remark~\ref{rem:asym}; this is the
trivial-irreducible-representation direction the averaging cannot touch). \textbf{Red:} matching that conserved scalar datum as
well leaves only $O(1)$ bulk modes. All slow rates scale linearly in the breaking $\eta$, confirming
Lemma~\ref{lem:schur}.}
\label{fig:cascadeB}
\end{figure}

\section{Scope and limitations}\label{sec:tight}

The construction possesses sharp boundaries. Defining these boundaries clarifies the operational scope of the framework.

\paragraph{Covariance of the breaking is essential.}
The structural backbone, Lemma~\ref{lem:cov}, of the cascade requires the full generator to be \(G\)-covariant. If the perturbation $\Lind_1$ breaks covariance instead of merely breaking conservation, $\Lind^\dagger$ stops functioning as an intertwiner. The isotypic blocks mix at order \(O(\eta)\), and the slow eigen-operators no longer carry definite irreducible representations.Consequently, the symmetry charges cease to be the slow modes, destroying the a priori warm start. The system reverts to the generic identification cost described in Section~\ref{sec:dichotomy}. Covariance defines the strict boundary between weakly broken symmetry where conservation is lifted but structure remains intact, and total symmetry loss.

\paragraph{Asymmetry targets only the gauge directions.}
The asymmetry measure \(\mathsf{A}_{G}\) vanishes if and only if \(\rho \) is \(G\)-invariant. It is blind to slow modes living in the trivial irreducible representation, which constitutes the commutant or multiplicity sector. When the representation carries multiplicities, \(U_g=\bigoplus_\lambda D_\lambda(g)\otimes\mathbb{I}_{m_\lambda}\), these modes correspond to logical operators on the \(m_{\lambda }\)-dimensional multiplicity spaces. This structure forms a noiseless subsystem.Any near-conserved logical datum remains a slow mode that minimizing \(\mathsf{A}_{G}\) cannot remove. The \(SU(2)\) example in Section~\ref{sec:example}\ documents this effect. The residual \(0.137\) scalar mode represents a conserved sector population, meaning symmetrization yields only a partial speedup from \(0.101\) to \(0.137\).A complete warm start requires satisfying two independent conditions: matching the asymmetry gauge directions and matching the thermal logical data of the commutant. The theorem resolves the first condition. The second requires matching a model-specific, finite set of expectations.

\paragraph{Perturbation order and rate clustering.}
The slow rate scales as \(g_\lambda=\eta q_\lambda+O(\eta^2)\), where \(q_{\lambda }\) is the Schur scalar of $\Lind_1^\dagger$ on \(V_{\lambda }\). Proposition~\ref{prop:sandwich}\ bounds this error term explicitly and establishes a global linear lower bound. Generically, \(q_\lambda>0\) and the splitting occurs at first order. If the leading covariant coupling annihilates a multiplet so that \(q_\lambda=0\), that multiplet relaxes at order \(O(\eta^2)\). This isolation enhances the rate jump while leaving the warm-start condition invariant.Conversely, the resolution of the cascade depends on the distinctness of the \(q_{\lambda }\) values. If two inequivalent irreducible representations share a rate, their subgroup steps merge. If all \(q_{\lambda }\) values cluster, the cascade degrades into multiple \(O(1)\)-ratio steps rather than a single large jump. By Theorem~\ref{prop:dichotomy}\, each step provides a genuine rate improvement only when the surviving rate is well separated from the subsequent lower rate.

\paragraph{Pinsker inequality is one-directional.}
The quantitative bound in Theorem~\ref{prop:nonabelian}(ii)\ controls the slow residual via \(\mathsf{A}_{G}\) through the quantum Pinsker inequality. Small asymmetry implies a small residual, but the converse is not assumed. Approximate symmetrization reduces the prefactor scale by \(\mathsf{A}_{G}^{1/2}\).As noted in Remark~\ref{rem:exact}\, this reduction becomes a true rate jump only when the residual is suppressed to the target tolerance. Approximate symmetry provides a quantitative advantage, but the qualitative rate jump requires exact invariance.

\section{Discussion and outlook}\label{sec:disc}
The cascade defines the structural signature separating a non-Abelian symmetry from an Abelian one and from a single conserved logical sector. It provides a lattice of partial warm starts, with one per subgroup, where each step clears a specific set of irreducible multiplets. Matching first moments is provably insufficient (Corollary~\ref{cor:firstmoments}). The initialization target is a resource-theoretic asymmetry rather than a standard checklist of expectation values. As the worked sampler demonstrates, this asymmetry composes with a separate matching of the conserved logical data to provide a full accounting of the slow manifold.

Two research threads extend from these results.

First, because \(\mathsf{A}_{G}\) and its subgroup residuals are estimable by the same randomized-measurement primitives that certify samplers, the quality of a warm start is measurable directly on the device. This ties the initialization protocol directly to self-certified halting criteria.

Second, the multiplicity and commutant structure that limits the asymmetry matches the mathematical structure of a quantum error-correcting code. It acts as a noiseless subsystem whose thermal logical content must be matched independently. Viewed from this perspective, the slow manifold represents a protected sector, making the cooling bottleneck dual to memory protection. This reveals a fundamental connection between dissipative sampling and quantum self-correction that will be developed in future work.

The group-averaging asymmetry governing the warm start in this framework is the exact same quantity that controls the entanglement-asymmetry quantum Mpemba effect in closed-system symmetry restoration \cite{AresMurcianoCalabrese}. The exact-nulling dichotomy proven here forms its dissipative analogue \cite{LuRaz}. The mechanism established in this work is purely spectral and algorithmic, an initialization strategy that intentionally places zero weight on the slow eigenspace, rather than a thermal-relaxation anomaly. The two distinct lines of research meet at the asymmetry as their shared control quantity.

\section*{Acknowledgements}
I thank Sebastian Diehl and Ruizhe Zhang for useful conversations. I am also grateful to Thiago Bergamaschi for a close reading of a draft and for detailed
comments. Among them, the trace-norm convention for mixing times, the practical status of strong
symmetries, and the suggestion that the perturbative gap analysis be complemented by a vectorized,
non-perturbative one via the concavity of the spectral gap \cite{Onorati} materially improved the present paper. 

Part of this work was performed while I was visiting the Institute for Theoretical Physics at Cologne University (Germany), the Institute for Pure and Applied Mathematics (IPAM, USA), the Simons Institute for the Theory of Computing (SIfTC, USA), and the Fields Institute for Research in Mathematical Sciences (FIRMS, Canada). This work was supported in part by the U.S. National Science Foundation and the Department of Energy through IPAM and SIfTC, by the Simons Foundation through SIfTC, and by the National Research Foundation of Korea (NRF) (RS-2021-NR060112) and Kwangwoon University. I used ChatGPT5.6 to write and execute the Python plotting script for Figure 1. I verified all plotted theoretical rates, codes and data.

\end{document}